\documentclass[hidelinks,onefignum,onetabnum]{siamart250211}

\usepackage{amsfonts,amsmath,amssymb}

\usepackage{comment}

\usepackage{tikz}
    \usetikzlibrary{decorations,decorations.pathmorphing, decorations.pathreplacing}   

\newcommand{\alg}[1]{\mathbf{#1}}
\newcommand{\algA}{\alg{A}}
\newcommand{\algB}{\alg{B}}
\newcommand{\algC}{\alg{C}}
\newcommand{\algL}{\alg{L}}
\newcommand{\algS}{\alg{S}}
\newcommand{\algU}{\alg{U}}
\newcommand{\N}{\mathbb N}
\newcommand{\Z}{\mathbb Z}

\DeclareMathOperator{\CSP}{\mathrm{CSP}}
\DeclareMathOperator{\Inv}{Inv}

\DeclareMathOperator{\Pol}{Pol}
\DeclareMathOperator{\Sig}{Sig}
\DeclareMathOperator{\Con}{Con}
\DeclareMathOperator{\HHH}{H}
\DeclareMathOperator{\SSS}{S}
\DeclareMathOperator{\PPP}{P}
\DeclareMathOperator{\HS}{HS}
\DeclareMathOperator{\HSfactors}{\mathcal F}
\DeclareMathOperator{\HSP}{HSP}
\DeclareMathOperator{\SMP}{SMP}
\DeclareMathOperator{\proj}{proj}
\DeclareMathOperator{\Sg}{Sg}
\DeclareMathOperator{\NP}{NP}
\DeclareMathOperator{\coNP}{co-NP}
\DeclareMathOperator{\dom}{dom}
\DeclareMathOperator{\Clo}{Clo}

\newtheorem{example}[theorem]{Example}
\newtheorem{conjecture}[theorem]{Conjecture}
\newtheorem{question}[theorem]{Question}

\headers{Polynomial definability in few subpowers}{J. Bulín and M. Kompatscher}

\title{Polynomial definability in constraint languages with few subpowers
    \thanks{
        A preliminary version was published in the proceedings of MFCS 2023 \cite{bulin_short_2023}.
        \funding{Both authors were supported by the Czech Science Foundation project 25-16324S and the MŠMT ČR INTER-EXCELLENCE project LTAUSA19070. The first author was supported by Charles University Research Centre UNCE/SCI/004. The second author was supported by grants PRIMUS/24/SCI/008 and UNCE/24/SCI/022 of Charles University.
        }
    }
}

\author{
    Jakub Bulín
        \thanks{Department of Theoretical Computer Science and Mathematical Logic, Faculty of Mathematics and Physics, Charles University, Prague, Czechia (\email{jakub.bulin@mff.cuni.cz}).}
    \and 
    Michael Kompatscher
        \thanks{Department of Algebra, Faculty of Mathematics and Physics, Charles University, Prague, Czechia (\email{kompatscher@karlin.mff.cuni.cz}).}
}

\begin{document}

\maketitle

\begin{abstract}
    A first-order formula is called \emph{primitive positive (pp)} if it uses only existential quantifiers and conjunctions. Primitive positive formulas are a central concept in (fixed-template) constraint satisfaction, as $\CSP(\Gamma)$ can be viewed as the problem of deciding the primitive positive theory of $\Gamma$, and pp-definability captures gadget reductions between CSPs. An important class of tractable constraint languages $\Gamma$ is characterized by the property of having \emph{few subpowers}, meaning that the number of $n$-ary relations pp-definable from $\Gamma$ is bounded by $2^{p(n)}$ for some polynomial $p(n)$. In this paper, we study a restriction of this property, namely that every pp-definable relation is definable by a pp-formula of polynomial length. We conjecture that the existence of such \emph{short definitions} is actually equivalent to $\Gamma$ having few subpowers, and we verify this conjecture for a large subclass, which in particular includes all constraint languages on three-element domains. Furthermore, we discuss how our conjecture imposes an upper complexity bound of $\coNP$ on the subpower membership problem for algebras with few subpowers.
\end{abstract}



\section{Introduction}

Constraint satisfaction provides a unifying framework for expressing a wide range of computational tasks arising from a smorgasbord of real-life applications and theoretical contexts. In a CSP instance, the goal is to assign values to variables subject to a list of \emph{constraints} to be satisfied. In the most general setting, a constraint consists of a tuple of variables (its \emph{scope}) and a list of admissible evaluations of the scope (i.e., tuples of values, forming the \emph{constraint relation}). Usually, the set of variables, the sets of admissible values for every variable (its \emph{domain}), and the list of input constraints are all finite. This simple formulation strikes a `perfect balance between generality and structure'~\cite{barto_polymorphisms_2017}.

In this general formulation, the CSP is an NP-complete problem: classical problems such as SAT or graph 3-colorability are easily expressible in this framework. Nevertheless, many problems subsumed by it are tractable, e.g., 2-SAT, Horn-SAT, or checking the consistency of a system of linear equations over $\mathbb Z_p$. A natural way to explore the complexity landscape of the CSP, justified by applications as well as theory~\cite{feder_computational_1998}, is to fix a finite domain $A$ and a finite set of relations $\Gamma$ on $A$ that are allowed to appear as constraints (a \emph{constraint language})~\cite{schaefer_complexity_1978}; the resulting fragment of the CSP is usually denoted by $\CSP(\Gamma)$. In the literature, sometimes also constraint languages on infinite domains, or with infinitely many relations are considered. But in this paper we will keep the two finiteness assumptions throughout.

The CSP dichotomy theorem~\cite{bulatov_dichotomy_2017,zhuk_proof_2017,zhuk_proof_2020} states that for every constraint language $\Gamma$, $\CSP(\Gamma)$ is in P or NP-complete. To tame the vast landscape of constraint languages, it was immensely helpful to realize that various ad hoc `gadget' complexity reductions share a common explanation using the notion of \emph{primitive positive} (\emph{pp-}) \emph{definability} (i.e., the usual first-order logic definability restricted to $\{\exists,\wedge,=\}$-formulas)~\cite{jeavons_closure_1997,jeavons_constraints_1998} and the more general notions of \emph{pp-interpretability} and \emph{pp-constructibility}~\cite{barto_reflections_2018}. For example, assuming $\mathrm{P}\neq\mathrm{NP}$, the CSP dichotomy theorem implies that $\CSP(\Gamma)$ is NP-complete if and only if $\Gamma$ pp-constructs \emph{every} finite constraint language. Moreover, pp-definability and its generalizations have an external characterization via the so-called \emph{polymorphisms} (`multivariate homomorphisms')~\cite{geiger_closed_1968,bodnarcuk_galois_1969,jeavons_algebraic_1998}. For an introduction to the area, we refer the reader to the survey~\cite{barto_polymorphisms_2017}.

Constraint languages $\Gamma$ for which $\CSP(\Gamma)$ is solvable by a certain algorithmic approach involving computing with \emph{compact representations} of solution sets (generalizing \emph{bases} of vector spaces or \emph{strong generating sets} over groups from the Schreier-Sims algorithm \cite{sims_computational_1970}) were characterized in~\cite{berman_varieties_2010,idziak_tractability_2010} as those that have \emph{few subpowers}, that is, the number of $n$-ary relations pp-definable from $\Gamma$ is bounded by $2^{p(n)}$ for some polynomial $p(n)$. This property, also called \emph{polynomial expressiveness} \cite{chen_expressive_2005}, is equivalent to having either of the following two properties, where \emph{small} means of size bounded by a polynomial in the arity $n$:
\begin{itemize}
    \item \emph{small generating sets}, i.e., every relation pp-definable from $\Gamma$ has a small subset that is not contained in any proper pp-definable subset,
    \item \emph{small independent sets}, i.e., sets of tuples such that every tuple can be separated from the remaining tuples by a pp-formula, are small.
\end{itemize}
The equivalence of these properties, inspired by basic linear algebra, was established in~\cite[Proposition 1.4]{berman_varieties_2010}.

In this paper, we study another measure of `smallness' of a constraint language: We say that a constraint language $\Gamma$ has \emph{short pp-definitions}, if every $n$-ary relation pp-definable from $\Gamma$ is definable by some primitive positive formula, whose lenght is polynomial in $n$. Examples include constraint languages encoding 2-SAT or the consistency of linear systems over $\mathbb Z_p$.

Note that there are $2^{O(p(n))}$ formulas of length at most $p(n)$, which implies that constraint language with short pp-definitions have few subpowers. We conjecture that also the converse is also true, and thus the two properties are equivalent:

\begin{conjecture}\label{conjecture:short-definitions-iff-few-subpowers}
    A constraint language has short pp-definitions if and only if it has few subpowers.
\end{conjecture}

Let us make two remarks here: By \cite[Theorem 3.12]{berman_varieties_2010} every constraint languages $\Gamma$ that does not have few subpowers, has $2^{2^{\Omega(n)}}$ many pp-definable relations of arity $n$. Therefore every such $\Gamma$ requires exponential-length pp-definitions, and the `only if' direction of Conjecture \ref{conjecture:short-definitions-iff-few-subpowers} holds.

Second, it can be easily seen that Conjecture \ref{conjecture:short-definitions-iff-few-subpowers} is true in the Boolean case, i.e., for constraint languages on a two-element domain. This fact was first stated in~\cite[Corollary~1]{lagerkvist_polynomially_2014} where the authors studied an equivalent condition (namely definability by pp-formulas with polynomially many existential quantifiers) for Boolean constraint languages, under the name \emph{polynomial closedness}. The reason is that Boolean constraint languages with few subpowers fall into one of two well-behaved types, which we will demonstrate in Examples~\ref{example:lin} and~\ref{example:2SAT}.

The main result of our paper, stated in Theorem~\ref{theorem:main-result}, confirms the conjecture for a substantial class of constraint languages, namely those whose \emph{polymorphism algebra} generates a \emph{residually finite variety}. This, in particular, implies that Conjecture~\ref{conjecture:short-definitions-iff-few-subpowers} also holds for constraint languages on three-element domains (Corollary \ref{cor:3}). The proof proceeds by first showing that it is enough to consider so-called \emph{critical} relations, and then employing structural theorems from universal algebra, in a similar fashion as in~\cite{bulatov_subpower_2019}.

Apart from being a natural property from the point of view of constraint satisfaction, in Section \ref{section:SMP} we argue that short pp-definitions have further applications in the study of the \emph{subpower membership problem} $\SMP(\algA)$ over an algebraic structure~$\algA$ (see \cite{kozen_complexity_1977,mayr_subpower_2012,bulatov_subpower_2019}), i.e., the problem of deciding whether a given list of tuples over $\algA$ generates another tuple over $\algA$.

The aforementioned compact representations provide a natural certificate for `Yes'-instances of $\SMP(\algA)$; this fact was used in \cite{bulatov_subpower_2019} to establish that $\SMP(\algA) \in \NP$. We argue that short pp-definitions can serve as a natural certificate for `No'-instances. In particular, we show how short pp-definitions impose an upper complexity bound of $\coNP$ on the subpower membership problem. Thus, Conjecture \ref{conjecture:short-definitions-iff-few-subpowers} would imply that $\SMP(\algA) \in \NP \cap \coNP$, for all algebras $\algA$ with few subpowers. Moreover, efficiently computable short definitions (see Question~\ref{question:ppp}) would put the problem $\SMP(\algA)$ in P.

\subsection*{Organization of the paper}
 
In Section~\ref{section:preliminaries} we give a formal definition of short pp-definitions, examples, and an exposition of related properties. Section~\ref{section:algebra} explains some necessary background from universal algebra. In particular, we introduce notions necessary to state our main result, and show that Conjecture \ref{conjecture:short-definitions-iff-few-subpowers} can be phrased in purely algebraic terms. In Section~\ref{section:main-result} we prove the main result, Theorem~\ref{theorem:main-result}. Section~\ref{section:SMP} introduces the connection to the Subpower Membership Problem. Finally, in Section~\ref{section:discussion}, we discuss open questions and possible directions of further research.

\section{Preliminaries} \label{section:preliminaries}

Let $A$ be a finite set. An $n$-ary relation $R$ on $A$ is any set of $n$-tuples $R\subseteq A^n$. By a \emph{constraint language} on the set $A$ (its \emph{domain}) we mean any finite set $\Gamma=\{R_1,\dots,R_m\}$ of relations on $A$ of arbitrary (but finite) arities.

We denote by $[n]$ the set $\{1,\dots,n\}$. For $\bar a\in A^n$ and any subset of coordinates $I\subseteq [n]$, where $I=\{i_1,\dots,i_k\}$ and $i_1<\dots<i_k$, the \emph{projection} of $\bar a$ to $I$ is the $k$-tuple $\proj_I \bar a=(a_{i_1},\dots,a_{i_k})$. Similarly, the \emph{projection} of $R\subseteq A^n$ to $I$ is the $k$-ary relation $\proj_I R=\{\proj_I\bar a\mid \bar a\in R\}$. For convenience we also write $\proj_i R$ instead of $\proj_{\{i\}} R$. The $i$th variable of $R$ is a \emph{dummy variable} if $R$ does not depend on the $i$th coordinate, i.e., $\bar a \in R$ if and only if $\proj_{[n]\setminus \{i\}} \bar a \in \proj_{[n]\setminus \{i\}} R$, for every $\bar a \in A^n$.

A relation $R$ is \emph{primitive positive definable} (or \emph{pp-definable} for short) from $\Gamma$, if it is definable in first-order logic by a formula using only the relations from $\Gamma$, the identity relation $=_A$ on $A$, conjunction, and existential quantification. Equivalently $R$ can be defined by a formula $\phi$ in prenex normal form
$$
\phi(x_1,\dots,x_n) \equiv \exists y_1 \exists y_2\dots \exists y_k \bigwedge_{i\in\{1,\dots,m\}}S_i(z_1^i,\dots,z_{r_i}^i),
$$
where $m\in \N$, every $S_i$ is an $r_i$-ary relational symbol representing a relation from $\Gamma\cup\{=_A\}$ and $z_j^i\in\{x_1,\dots,x_n,y_1,\dots,y_k\}$.

The set of all relations pp-definable from $\Gamma$, denoted by $\langle\Gamma\rangle$, forms a \emph{relational clone}, i.e., a set of relations on $A$ containing the identity relation and closed under intersections, direct products, projections to subsets of coordinates, and permutations of coordinates. Any constraint language $\Gamma$ that generates a relational clone $\mathcal R = \langle\Gamma\rangle$ is called a \emph{relational basis} of $\mathcal R$. Let us denote by $\langle\Gamma\rangle_n$ the set of all $n$-ary relations pp-definable from $\Gamma$. 

Note that the constraint satisfaction problem $\CSP(\Gamma)$ can be viewed as the problem of deciding the primitive positive fragment of the first-order theory of the relational structure $\mathbb A=(A;\Gamma)$. Relations $R\in\langle\Gamma\rangle$ then correspond to projections of solution sets of instances of $\CSP(\Gamma)$.

But the usefulness of pp-definability for the CSP goes much further than that. A key observation is summarized in the following theorem going back to \cite{jeavons_closure_1997}. 

\begin{theorem}
    If $\Gamma$ and $\Delta$ are constraint languages such that $\Delta\subseteq\langle\Gamma\rangle$, then there is a logspace reduction from $\CSP(\Delta)$ to $\CSP(\Gamma)$.
\end{theorem}

The proof is straightforward; the reduction is obtained by replacing each constraint of an instance of $\CSP(\Delta)$ by its pp-definition from $\Gamma$. In this paper, we are interested in the size of this pp-definition.

More generally, also so called \emph{pp-interpretations} and \emph{pp-constructions} induce logspace reduction between CSPs on possibly different domains. For a constraint languages $\Gamma$ on set $A$ and $\Delta$ on set $B$, we say that a partial map $I \colon A^n \to B$ is a \emph{pp-interpretation} of $\Delta$ in $\Gamma$, if the domain $\dom(I) \subseteq A^n$, the kernel $\ker(I) = \{ (\bar x,\bar y) \mid I(\bar x) = I(\bar y) \} \subseteq A^{2n}$, as well as the pre-image $I^{-1}(R)$ of every relation $R \in \Delta$ are all pp-definable in $\Gamma$. We say that $\Delta$ is \emph{pp-interpretable} in $\Gamma$ if there is a pp-interpretation of $\Delta$ in $\Gamma$, and $\Delta$ is  \emph{pp-constructible} from $\Gamma$ if it can be obtained from $\Gamma$ by pp-interpretations and homomorphic equivalence. In both cases there is a logspace reduction from $\CSP(\Delta)$ to $\CSP(\Gamma)$.

We say that $\Gamma$ and $\Delta$ are \emph{pp-bi-interpretable}, if there are pp-interpretations $I_1$ of $\Gamma$ in $\Delta$ and $I_2$ of $\Delta$ in $\Gamma$, such that the function graphs of their compositions $I_2\circ I_1$ and $I_1\circ I_2$ are pp-definable in $\Delta$ and $\Gamma$ respectively. For a broader exposition of these concepts see \cite{barto_polymorphisms_2017} and \cite{barto_reflections_2018}.

\subsection{Few subpowers and short definitions}

In order to put Conjecture~\ref{conjecture:short-definitions-iff-few-subpowers} on a firm footing, let us next formally define the notion of \emph{few subpowers}~\cite{berman_varieties_2010,idziak_tractability_2010} and the central concept of the present paper, \emph{short pp-definitions}.

\begin{definition}
    A constraint language $\Gamma$ has \emph{few subpowers}, if there exists a polynomial $p(n)$ such that $|\langle\Gamma\rangle_n|\leq 2^{p(n)}$ for all $n>0$.
\end{definition}

\begin{definition}\label{definition:short-pp-def}
    Let $\Gamma$ be a constraint language and $f \colon \N \to \N$ a monotone function. We say that $\Gamma$ has:
    \begin{itemize}
        \item \emph{pp-definitions of length $f(n)$}, if for every $n>0$ and every $R\in\langle\Gamma\rangle_n$, $R$ is definable from $\Gamma$ by a primitive positive formula $\phi$ of length $|\phi|\leq f(n)$.
        \item \emph{short pp-definitions}, if $\Gamma$ has pp-definitions of length $p(n)$ for some polynomial $p(n)$.
    \end{itemize}
\end{definition}

Here, we consider the length $|\phi|$ to be simply the number of symbols in some syntactical representation of the formula. In the definition of \emph{short pp-definitions}, one could alternatively impose a polynomial bound on either
\begin{itemize}
    \item the number of atomic subformulas in $\phi$, or
    \item the number of existentially quantified variables of $\phi$.
\end{itemize}
The latter option was used in \cite{lagerkvist_polynomially_2014} in the notion of \emph{polynomial closedness} of $\langle\Gamma\rangle$. Note that, since $\Gamma$ is fixed and finite, the three definitions coincide.

Clearly, having few subpowers is a property of the relational clone $\mathcal R=\langle\Gamma\rangle$, independent of the choice of the relational basis $\Gamma$. We observe that the same is true for short pp-definitions. In fact, up to multiplication by a scalar, this is true for any bound $f(n)$ on the length of pp-definitions:

\begin{lemma} \label{lemma:ppinterdef}
    Let $\Gamma$ and $\Delta$ be constraint languages such that $\langle\Gamma\rangle=\langle\Delta\rangle$. If $\Gamma$ has pp-definitions of length $f(n)$, then $\Delta$ has pp-definitions of length $O(f(n))$. In particular, $\Gamma$ has short pp-definitions if and only if $\Delta$ does.
\end{lemma}

\begin{proof}
    Let $R\in\langle\Delta\rangle_n=\langle\Gamma\rangle_n$. By assumption, $R$ has a pp-definition $\phi_R$ from $\Gamma$ of length at most $f(n)$. Since $\Gamma\subseteq\langle\Gamma\rangle=\langle\Delta\rangle$, every relation $S\in\Gamma$ can be defined from $\Delta$ by some pp-formula $\psi_S$. Let $c=\max\{|\psi_S|\colon S\in\Gamma\}$. If we replace every atomic subformula $S_i(z_1^i,\dots,z_{r_i}^i)$ in $\phi_R$ by a suitable variant of the formula $\psi_{S_i}$ (renaming variables as needed), we obtain a pp-definition of $R$ from $\Delta$ of length at most $c\cdot f(n)$.
\end{proof}

So, by Lemma \ref{lemma:ppinterdef}, it makes also sense to say that a relational clone $\mathcal R$ with finite relational basis has short pp-definitions (or pp-definitions of length $O(f(n))$), as this property is independent of the choice of the constraint language $\Gamma$ with $\mathcal R = \langle \Gamma \rangle$. We remark that, by the main result of \cite{aichinger_number_2014}, all relational clones with few subpowers have a finite relational basis; however this is not true for relational clones in general.

We further remark that there is no constraint language $\Gamma$ on sets $A$ with $|A|>1$ that has pp-definitions of sub-linear, or even linear length. For this simply note that the equality relation $=$ on $A$ already pp-defines $B_n$-many $n$-ary relations, where $B_n$ denotes the $n$-th Bell number, i.e., the number of possible partitions of a set of size $n$. Thus $\Omega(\log (B_n))$ is a lower bound on the length of pp-definitions in any non-trivial constraint language. By an easy bars-and-stars argument, one can see that $B_n\geq 2^{n-1}$, thus pp-definitions require at least linear length. In fact, from the asymptotics $\log (B_n) = n (\log n - \log \log n - 1 + o(1))$ (see e.g. \cite{debruijn_asymptotic_1981}), we can derive the lower bound $\Omega(\log (B_n)) = \Omega(n \log(n))$ for $n > 1$.

We next generalize Lemma \ref{lemma:ppinterdef} to constraint languages on different domains as follows:

\begin{lemma} \label{lemma:ppbiinterpret}
    Let $\Gamma$ and $\Delta$ be a pair of pp-bi-interpretable constraint languages. If $\Gamma$ has pp-definitions of length $f(n)$, then $\Delta$ has pp-definitions of length $O(f(cn))$, for some $c\geq 1$. In particular, $\Gamma$ has short pp-definitions if and only if $\Delta$ does.
\end{lemma}

\begin{proof}
Let $A$ be the domain of $\Gamma$ and $B$ the domain of $\Delta$. First note that the statement trivially holds, if $|A| = |B|=1$, as then both have constant length pp-definitions.

In the remaining case, both $|A| > 1$ and $|B|>1$ must hold, as otherwise $\Gamma$ and $\Delta$ could not be pp-bi-interpretable. Let $I_1 \colon A^{k_1} \to B$ be a pp-interpretation of $\Delta$ in $\Gamma$ and $I_2 \colon B^{k_2} \to A$ a pp-interpretation of $\Gamma$ in $\Delta$, such that the function graphs of $I_1\circ I_2 \colon B^{k_2 k_1} \to B$ and $I_2\circ I_1 \colon A^{k_1 k_2} \to A$ are pp-definable in $\Delta$ and $\Gamma$ respectively.

Let $R \in \langle\Delta\rangle_n$. The $k_1n$-ary relation $I_1^{-1}(R)$ is an element of $\langle \Gamma \rangle_{k_1n}$, and has therefore a pp-definition $\phi$ from $\Gamma$ of length at most $f(k_1n)$. 

Also, since $I_2$ is a pp-interpretation, we know that $\dom(I_2) \subseteq A^{k_2}$ and $I_2^{-1}(S)$ for every relation $S\in\Gamma \cup \{=_A\}$ have pp-definitions in $\Delta$. Thus, we obtain a pp-definition $\psi$ of $I_2^{-1}I_1^{-1}(R)$ in $\Delta$, by substituting every variable $z$ in $\phi$ by a tuple of $k_2$ many variables $\bar z = (z_1,\ldots,z_{k_2}) \in \dom(I_2)$, and every atomic subformula $S(z_1,\dots,z_{r_i})$ of $\phi$, by the corresponding pp-definition of $I_2^{-1}(S)(\bar z_1,\dots,\bar z_{r_i})$. Note that $\psi$ is of length $O (k_2f(k_1n)) = O (f(k_1n))$. 

A pp-definition of $R(x_1,\ldots,x_n)$ can then be obtained from 
$$
\exists \bar y_1,\ldots,\bar y_{n} (\psi(\bar y_1,\ldots,\bar y_{n}) \land \bigwedge_{i=1}^n I_2(I_1(\bar y_i)) = x_i),
$$ 
which is of length $O (f(k_1n) + dn)$, where $d$ is the length of the pp-definition of the function graph of $I_2\circ I_1$. Since $f(n) \in \Omega (n)$ we get $O (f(k_1n) + dn) = O (f(k_1n))$, which finishes the proof of the first statement.

The statement about short pp-definitions immediately follows from applying the result to polynomial growth functions $f(n)$.
\end{proof}

We remark that a similar statement does not hold for structures that pp-construct each other, as it is an easy exercise to construct a constraint language without few subpowers that is homomorphically equivalent to $\Gamma_{\mathrm{Lin}}$, a constraint language we prove to have short pp-definitions in Example \ref{example:2SAT}.

\subsection{Polymorphisms}

Central to the algebraic approach to the CSP is the idea that constraint languages up to pp-definability (that is, relational clones) can be characterized by their \emph{polymorphisms}. Following the terminology from \cite{barto_polymorphisms_2017}, a $k$-ary operation $f:A^k\to A$ is \emph{compatible} with an $n$-ary relation $R\subseteq A^n$, and $R$ is \emph{invariant} under $f$, if for all tuples $\bar a_1,\bar a_2\ldots,\bar a_k \in R$ also $f(\bar a_1,\bar a_2\ldots,\bar a_k) \in R$, where the image $f(\bar a_1,\bar a_2\ldots,\bar a_k) \in A^n$ is computed coordinate-wise. A \emph{polymorphism} of a constraint language $\Gamma$ is then any function on the domain that is compatible with all relations from $\Gamma$. As is usual, we write $\Pol(\Gamma)$ to denote the set of all polymorphisms of $\Gamma$ and, similarly, $\Inv(\mathcal F)$ for the set of all relations on the domain $A$ invariant under a set of operations $\mathcal F$. It is well known that $\Pol(\Gamma)$ forms an algebraic object called a \emph{clone}, i.e. it is closed under composition, and contains all projections. The key connection between polymorphisms and pp-definability can be summarized in the following lemma.

\begin{lemma}[\cite{geiger_closed_1968,bodnarcuk_galois_1969,jeavons_algebraic_1998}]
    For any constraint language $\Gamma$, $\langle\Gamma\rangle=\Inv(\Pol(\Gamma))$.
\end{lemma}

Constraint languages with few subpowers can be characterized by the existence of an \emph{edge polymorphism}, that is, a polymorphism satisfying certain algebraic identities (under all evaluations of variables in the domain).

\begin{theorem}[\cite{berman_varieties_2010, idziak_tractability_2010}]\label{theorem:few_subpowers_iff_edge_term}
    A constraint language $\Gamma$ has few subpowers, if and only if for some $k\geq 2$ there exists a \emph{$k$-edge polymorphism} $e\in\Pol(\Gamma)$, that is, a $(k+1)$-ary operation $e\colon A^{k+1}\to A$ satisfying the following identities:
    \begin{align*}
        e(y,y,x,x,x,\dots,x) &\approx x\\
        e(y,x,y,x,x,\dots,x) &\approx x\\
        e(x,x,x,y,x,\dots,x) &\approx x\\
        e(x,x,x,x,y,\dots,x) &\approx x\\
        &\vdots\\
        e(x,x,x,x,x,\dots,y) &\approx x
    \end{align*}
    In that case, the problem $\CSP(\Gamma)$ is in $\mathrm{P}$.
\end{theorem}
The following two special cases were important intermediate steps towards Theorem~\ref{theorem:few_subpowers_iff_edge_term} and, in particular, cover all Boolean constraint languages with few subpowers:
\begin{itemize}
    \item A \emph{Mal'tsev} operation is a ternary operation $m\colon A^3\to A$ satisfying the identities $m(x,x,y)\approx m(y,x,x)\approx y$. If $m$ is a Mal'tsev operation, then $e(x_1,x_2,x_3)=m(x_2,x_1,x_3)$ is a 2-edge term.
    \item A \emph{near-unanimity} operation (of arity $k\geq 3$) is an operation $t$ satisfying the following identities:
    $$
    x\approx t(y,x,\dots,x)\approx t(x,y,x,\dots,x)\approx t(x,\dots,x,y,x)\approx t(x,\dots,x,y)
    $$ 
    Then $e(x_1,x_2,\dots,x_{k+1})=t(x_2,\dots,x_{k+1})$ is a $k$-edge term. A ternary near-unanimity is called a \emph{majority}.
\end{itemize}

\subsection{Two canonical examples}

Let us now give two examples of constraint languages with short pp-definitions; Example~\ref{example:lin} is invariant under a Mal'tsev operation, while Example~\ref{example:2SAT} admits a majority polymorphism.

\begin{example} \label{example:lin}
    The problem of checking consistency of a linear system over $\mathbb Z_2$ can be encoded as $\CSP(\Gamma_{\mathrm{Lin}})$ for $\Gamma_{\mathrm{Lin}}=\{R_{\mathrm{Lin}},C_0,C_1\}$ where 
    $$
    R_{\mathrm{Lin}}=\{(0,0,0),(0,1,1),(1,0,1),(1,1,0)\}
    $$ 
    consists of all triples satisfying the linear equation `$x_1+x_2=x_3$', $C_0=\{0\}$, and $C_1=\{1\}$. Indeed, any linear equation $x_1+x_2+\dots+x_n=b$ for $n>1$ can be encoded using auxiliary variables $y_1,\dots,y_{n-1}$ and three-variable equations
    $$
    x_1+x_2=y_1,\quad y_1+x_3=y_2,\quad \dots\quad y_{n-2}+x_n=y_{n-1},\quad y_{n-1}=b
    $$
    thus providing a pp-definition:
   \begin{equation}
    \exists y_1 \dots \exists y_{n-1} \left( R_{\mathrm{Lin}}(x_1,x_2,y_1)\wedge\dots\wedge R_{\mathrm{Lin}}(y_{n-2},x_n,y_{n-1})\wedge C_b(y_{n-1}) \right) \label{eq:linformula}
\end{equation}
    (For $n=1$, the equation $x_1=b$ is pp-defined simply by $C_b(x_1)$.)

    Thus the relational clone $\langle\Gamma_{\mathrm{Lin}}\rangle$ consists of all relations definable via systems of linear equations, that is, all affine subspaces of $\mathbb Z_2^n$, for any $n>0$. It is well-known that every affine subspace $R$ of $\mathbb Z_2^n$ can be described by at most $n$ linear equations. The conjunction of the corresponding pp-formulas of the form \eqref{eq:linformula} clearly defines $R$ and is of length $O(n^2)$. Thus $\Gamma_{\mathrm{Lin}}$ has pp-definitions of length $O(n^2)$.
\end{example}

The constraint language $\Gamma_{\mathrm{Lin}}$ is arguably one of the easiest examples of a constraint language with a Mal'tsev polymorphism. Indeed, $\langle\Gamma_{\mathrm{Lin}}\rangle=\Inv(\{m\})$ for the Mal'tsev operation $m(x,y,z)=x-y+z \bmod 2$.

While Conjecture~\ref{conjecture:short-definitions-iff-few-subpowers} is open even under the presence of a Mal'tsev polymorphism, the above example can be generalized to constraint languages with a \emph{central} Mal'tsev polymorphism. This means that the Mal'tsev polymorphism $m$ is compatible with its own function graph, i.e., the 4-ary relation $R=\{(x,y,z,m(x,y,z))\mid x,y,z\in A\}$ is also pp-definable in $\Gamma$. In this case we know by \cite{gumm_algebras_1979} that the corresponding polymorphism clone $\Pol(\Gamma)$ is \emph{affine}, i.e., up to composition with constants it consists of all affine operations over a module. In particular, then $m(x,y,z) = x-y+z$, where $+$ is the addition in the underlying module. For such affine cases we get the following generalization of Example \ref{example:lin}:

\begin{theorem} \label{theorem:affine}
    Let $\Gamma$ be a constraint language such that $\Pol(\Gamma)$ is affine. Then $\Gamma$ has pp-definitions of length $O(n^2)$.
\end{theorem}

We remark that, at least for domains of prime size $p$, Theorem \ref{theorem:affine} can be proved analogously to Example \ref{example:lin}, since then $\langle \Gamma \rangle$ can be easily described as a nice subset of the relational clone of all affine subspaces of powers of $\Z_p$ (see e.g. \cite{bagyinski_lattice_1982}). But, as Theorem \ref{theorem:affine} is also a corollary of our main result (Theorem \ref{theorem:main-result}) we will not include a full proof of it at the current point.

For the second example, we need the following characterization of relations invariant under near-unanimity operations:

\begin{theorem}[\cite{baker_polynomial_1975}]\label{theorem:NU}
    If a relation $R\subseteq A^n$ is invariant under a $(k+1)$-ary near-unanimity operation $t$, then it is pp-definable from its projections to at most $k$-ary subsets of coordinates by the formula
    $$
    \bigwedge_{\substack{I = \{i_1,\ldots,i_l\} \\ I \subseteq [n], |I|\leq k}}\proj_I R(x_{i_1},\dots,x_{i_l}).
    $$
\end{theorem}

Thus, if a constraint language $\Gamma$ has a $(k+1)$-ary near-unanimity polymorphism $t$, then every relation $R \in \langle \Gamma \rangle_n$ (for $n>k$) can be written as the conjunction of $\binom{n}{k}$ relations of arity $k$. As a direct corollary of Theorem \ref{theorem:NU} we obtain:

\begin{corollary} \label{corollary:NU}
    Let $\Gamma$ be a constraint language which admits a $(k+1)$-ary near-unanimity polymorphism. Then $\Gamma$ has pp-definitions of length $O(n^k)$.
\end{corollary}

\begin{example} \label{example:2SAT}
    Corollary \ref{corollary:NU} can be exemplified by the the well-known 2-SAT problem. The standard way to encode 2-SAT as a CSP is by using the constraint language 
    $$
    \Gamma_{\text{2-SAT}}=\{R_{00},R_{01},R_{10},R_{11}\}
    $$ 
    where $R_{ij}=\{0,1\}^2\setminus\{(i,j)\}$ represent each clause type (see \cite[Example 2.2]{barto_polymorphisms_2017}). 
    
    It is well known that the relations pp-definable from $\Gamma_{\text{2-SAT}}$ are exactly those invariant under the majority operation, i.e., the unique 3-ary near-unanimity operation on $\{0,1\}$. By Corollary \ref{corollary:NU}, $\Gamma_{\text{2-SAT}}$ has pp-definitions of length $O(n^2)$.
\end{example}

Based on the above, it is easy to see that Conjecture~\ref{conjecture:short-definitions-iff-few-subpowers} holds for constraint languages $\Gamma$ over Boolean domains:

\begin{theorem}[see {\cite[Corollary~1]{lagerkvist_polynomially_2014}}] \label{theorem:Boolean}
    Let $\Gamma$ be a Boolean constraint language with few subpowers. Then $\Gamma$ has pp-definitions of length $O(n^2)$. 
\end{theorem}

\begin{proof}
    By the classification of Post's lattice \cite{post_two_1941}, every Boolean constraint language with few subpowers has either a majority polymorphism or an affine polymorphism clone. In both cases, we obtain quadratic-length pp-definitions, in the same fashion as in Examples \ref{example:lin} and \ref{example:2SAT}.
\end{proof}

However, on bigger domains, the situation is much more complicated. Already on the 3-element domain there are constraint languages that have few subpowers but neither a Mal'tsev, nor a near-unanimity polymorphism. For example \cite[Examples 2.1.1 and 2.1.2]{brady_notes_2025} have both a 3-edge polymorphism (and, in fact, even so-called \emph{generalized majority-minority polymorphism} \cite{dalmau_generalized_2006}), without being invariant under any Mal'tsev or near-unanimity operation.

\section{Universal-algebraic tools and constructions} \label{section:algebra}

In the following, we are going to briefly introduce some basic notions from universal algebra that will allow us to state our main result, Theorem \ref{theorem:main-result}, in its full generality. In Sections \ref{section:algshort} and \ref{section:multi-sorted}, we furthermore discuss how short pp-definitions behave with respect to basic algebraic constructions. For more background in universal algebra, we refer the reader to the standard textbooks \cite{bergman_universal_2011,burris_course_1981}.

An \emph{algebra} $\algA = (A;(f_i)^{\algA}_{i \in I})$ is a first-order structure in a purely functional language $(f_i)_{i \in I}$ (where each symbol $f_i$ has an associated \emph{arity}). We say $\algA$ is finite if its domain $A$ is finite.  A \emph{subalgebra} $\algB = (B;(f_i)^{\algB}_{i \in I})$ of an algebra $\algA = (A;(f_i)^{\algA}_{i \in I})$ (denoted $\algB \leq \algA$) is an algebra obtained by restricting all \emph{basic operations} $f_i^{\algA}$ to an invariant subset $B \subseteq A$. The \emph{product} $\prod_{i \in I} \algA_i$ of a family of algebras $(\algA_i)_{i \in I}$ in the same language is defined as the algebra with domain $\prod_{i\in I} A_i$, whose basic operations are defined coordinate-wise. A \emph{homomorphism} $h\colon \algA \to \algB$ between algebras is defined as a map that preserves all basic operations, i.e., $h(f_i^{\algA}(a_1,\ldots,a_n)) = f_i^{\algB}(h(a_1),\ldots,h(a_n))$ for all $i\in I$. 

The \emph{kernel} of any homomorphism (i.e., the relation defined by $(x,y)\in\theta \leftrightarrow h(x)=h(y)$) is a \emph{congruence}, that is, an equivalence relation invariant under $\algA$. Conversely, for every congruence $\alpha$ of $\algA$, it is easy to see that one can construct a \emph{quotient algebra} $\algA/\alpha$, as the homomorphic image of the quotient mapping $x \mapsto x/\alpha$. Under the inclusion order, the set of all congruences of an algebra $\algA$ forms the \emph{congruence lattice} $\Con(\algA)$. The minimal and maximal element of this lattice are always the trivial congruences $0_\algA = \{(x,x) \mid x \in A\}$ and $1_\algA = \{(x,y) \mid x,y \in A\}$. An algebra $\algA$ is called \emph{subdirectly irreducible} if $0_\algA$ has a unique cover in $\Con(\algA)$, i.e., there is a unique minimal non-trivial congruence (called the \emph{monolith}).

By $\HHH$, $\SSS$, and $\PPP$ we denote the closure of a set of algebras under homomorphic images, subalgebras, and products respectively. It is well-known that the closure of any set of algebras under $\HSP$ is a \emph{variety}, i.e., a class of algebras defined by a set of identities (by Birkhoff's theorem, see, e.g., \cite{bergman_universal_2011}). A variety is called \emph{residually finite} if (up to isomorphism) it only contains finitely many subdirectly irreducible algebras, all of which are finite.

By a \emph{factor} of an algebra $\algA$ we mean any quotient algebra of a subalgebra of~$\algA$. Let us denote the set of all factors of an algebra $\alg A$ by $\HSfactors(\algA)$. Note that $\HSfactors(\algA)$ is finite for finite algebras $\algA$, and every member of $\HS(\algA)$ (homomorphic images of subalgebras of $\algA$) is isomorphic to a factor $\algB\in\HSfactors(\alg A)$.

\subsection{Algebras with short pp-definitions} \label{section:algshort}

If we assign a function symbol to every element of $\Pol(\Gamma)$, for a constraint language $\Gamma$, then we can also regard $\algA = (A;\Pol(\Gamma))$ as an algebra (the \emph{polymorphism algebra} of $\Gamma$). On the other hand, for every algebra $\algA$, its invariant relations $\Inv(\algA)$ form a relational clone. Thus, it makes sense to say that a finite algebra $\algA$ has \emph{few subpowers} if $\Inv(\algA)$ has few subpowers.

Note that a relation $R$ is invariant under $\algA$ if and only if $R \leq \algA^n$ for some $n$, i.e., $R$ is a subalgebra of some finite power of $\algA$. (Such $R$ is also called a \emph{subpower of $\algA$}, from which the notion of having ``few subpowers'' derives its name.)

As mentioned earlier, by a (non-constructive) proof of Aichinger, Mayr and McKenzie \cite{aichinger_number_2014}, for every algebra $\algA$ with few subpowers there exists a finite relational basis of $\Inv(\algA)$, i.e., a constraint language $\Gamma$ such that $\Inv(\algA) = \langle \Gamma \rangle$ (for general algebras $\algA$ this is not always the case).

\begin{definition}
    An algebra $\algA$ has \emph{pp-definitions of length $O(f(n))$}, if there exists a constraint language $\Gamma$ such that $\Inv(\algA) = \langle \Gamma \rangle$, and $\Gamma$ has pp-definitions of length $O(f(n))$. An algebra $\algA$ has \emph{short pp-definitions}, if $\Inv(\algA) = \langle \Gamma \rangle$ for some $\Gamma$ with short pp-definitions.
\end{definition}

By Lemma \ref{lemma:ppinterdef}, having short pp-definitions (or definitions of length $O(f(n))$) is independent of the choice of the relational basis $\Gamma$. By \cite{geiger_closed_1968} we know that for a finite algebra $\algA$, the clone $\Clo(\algA) = \Pol(\Inv(\algA))$ consists exactly of all term operations of $\algA$, and $\Inv(\algA) = \Inv(\algB)$ if and only if $\Clo(\algA) = \Clo(\algB)$. Thus having short pp-definitions (or pp-definition of length $O(f(n))$) is, in fact, a property that only depends on the term clone $\Clo(\algA)$ of an algebra.

Moreover, we know that two clones on a finite set are isomorphic (i.e. there is a bijection between them preserving composition and every projection), if and only if their invariant relations are pp-bi-interpretable (see e.g. \cite{bodirsky_topological_2015}). Thus Lemma \ref{lemma:ppinterdef} implies that having short pp-definitions is even invariant under clone isomorphisms. Thus, one can argue that having short pp-definitions is a property of clones.

Nevertheless, we are going to present our proof in the language of algebras, not clones, as this will simplify the presentation significantly, especially given the use of multi-sorted approach that we will introduce in Section \ref{section:multi-sorted}.

Let us first observe the following fact about powers of algebras:

\begin{lemma} \label{lemma:poweralg}
    Let $\algA$ be an algebra and $\algB = \algA^k$ for some $k > 1$. Then $\algB$ has pp-definitions of length $O(f(n))$ if and only if $\algA$ has pp-definitions of length $O(f(\lceil \frac{n}{k} \rceil ))$.
\end{lemma}

We remark that any relational bases of $\algA$ and $\algA^k$ are pp-bi-interpretable. Thus, their equivalence regarding short pp-definitions already follows from Lemma \ref{lemma:ppbiinterpret}. In order to obtain the more precise statement of Lemma \ref{lemma:poweralg}, we however need to actually study these interpretations in more detail.

\begin{proof}[Proof of Lemma \ref{lemma:poweralg}]
    It is straightforward to see that any relation $R \subseteq B^n$ is invariant under $\algB$ if and only if it is invariant under $\algA$, when interpreted as an $kn$-ary relation on $A$.
    
    We are first going to prove the ``only if'' direction. Let $\Delta$ be a relational basis of $\Inv(\algB)$. By interpreting every $m$-ary relation $Q \in \Delta$ as a $km$-ary relation $Q' \leq \algA^{km}$, we obtain a relational basis $\Delta' = \{ Q' \mid Q \in \Delta \}$ of $\Inv(\algA)$. 
    
    Let $R\leq\algA^n$. By adding dummy variables, we can assume without loss of generality that $n=k\ell$ for $\ell = \lceil \frac{n}{k} \rceil$. Then by assumption, $R$ considered as an $\ell$-ary relation over $B$ has a pp-definition $\phi(x_1,\dots,x_\ell)$ from $\Delta$ of length in $O(f(\ell))$. If we substitute each ($B$-valued) variable in $\phi$ by a $k$-tuple of ($A$-valued) variables, and each  $\Delta$-predicate $Q$ in $\phi$ by the corresponding $\Delta'$-predicate $Q'$, then we obtain a pp-definition $\phi'$ of $R$ of length at most $k \cdot f(l) = k \cdot f(\lceil \frac{n}{k} \rceil)$. Note further that existentially quantifying all the additionally added dummy variables adds only constantly many symbols. Thus $\algA$ has pp-definitions of length $O(k \cdot f(\lceil \frac{n}{k} \rceil)) = O(f(\lceil \frac{n}{k} \rceil))$.
    
    Now let us prove the ``if'' direction. Let $\Gamma$ be a relational basis of $\Inv(\algA)$. Let $e \colon A \to A^k$ be the map $x \mapsto (x,\ldots,x)$. For every $Q \in \Gamma$, we define 
    $$
    Q' = \{(e(x_1),\ldots,e(x_m)) \mid (x_1,\ldots,x_m) \in Q\} \leq \algB^m$$ and also the binary relations $P_i = \{((x_1,\ldots,x_k), e(x_i)) \mid (x_1,\ldots,x_k) \in B \} \leq \algB^2$, for $i = 1,\ldots,k$. We construct the relational basis of $\Inv(\algB)$ as
    $$
    \Gamma'=\{Q'\mid Q\in\Gamma\}\cup\{P_1,\dots,P_k\}.
    $$

    We need to show that $\Gamma'$ and thus also $\alg B$ has pp-definitions of length $O(f(n))$.
    Let $R'\leq \algB^{n}$ and let $R\leq\alg A^{kn}$ be the relation $R'$ considered as a $kn$-ary relation over~$A$. By assumption, there exists a pp-definition $\phi(x_1,\ldots,x_{nk})$ of $R \leq \algA^{kn}$ over $\Gamma$  of length in $O(f(\lceil \frac{kn}{k}\rceil))=O(f(n))$. Let $z_1,\ldots,z_\ell$ be its existentially quantified variables. Let us then define $\phi'(y_1,y_2,\ldots,y_n)$ as a pp-formula over $\Gamma'$ with existentially quantified variables $x_1',\ldots,x_{nk}', z_1',\ldots,z_{\ell}'$ and predicates $P_i(y_j,x_{(j-1)k+i}')$ for all $i \in [k], j \in [n]$, as well as $Q'(u_1',\ldots,u_m')$ for every predicate $Q(u_1,\ldots,u_m)$ in $\phi$ with $u_i \in \{x_1,\ldots,x_{nk}, z_1,\ldots,z_\ell\}$. It is easy to check that $\phi'$ defines $R' \leq \algB^{n}$ over $\Gamma'$. Clearly, the length of $\phi'$ is in $O(f(n))$.
\end{proof}

\subsection{Multi-sorted view}\label{section:multi-sorted}

In the following, we will also make use of multi-sorted relations, as they provide a more natural framework for our proof in Section \ref{section:main-result}. To be more specific, if $\mathcal F$ is a finite family of finite algebras in the same language, then it still makes sense to define the set of all invariant relations $\Inv(\mathcal F)$ to consist of all relations $R \leq \algA_1 \times \ldots \times \algA_n$ for $\algA_1,\ldots,\algA_n \in \mathcal F$. In particular, the set of all such relations is still closed under intersections, direct products, projections and permutations of coordinates, and contains the equality relations on all possible domains. It is therefore a (multi-sorted) relational clone. Such relational clones can again be characterized as sets of relations that are closed under pp-definitions; the only difference to the single sorted case is that for each variable we also need to specify its sort/domain.

If furthermore all elements of $\mathcal F$ have a common $k$-edge term, the finite relational basis result of Aichinger, Mayr and McKenzie \cite{aichinger_number_2014} still applies, i.e. there is a finite set of (multi-sorted) relations $\Gamma$ such the relations invariant under $\mathcal F$ are exactly the relations pp-definable from $\Gamma$. The translation of the proof to the multi-sorted setting is straightforward; we refrain from giving the technical details here. Therefore it makes also sense to define the property of having short pp-definitions for such families of algebras $\mathcal F$.

We note that studying constraint languages in which the variables can range over different domains is a fairly standard viewpoint in CSP (as well as in clone theory); for instance, it was used in Zhuk's proof of the CSP dichotomy theorem \cite{zhuk_proof_2020}. 

This multi-sorted approach allows us to consider (multi-sorted) relations over the finite family $\HSfactors(\algA)$ of factors of $\algA$ rather than (single-sorted) relations over $\algA$ itself. This is justified by the following lemma.

\begin{lemma} \label{lemma:HS}
    Let $\alg A$ be a finite algebra. Then the algebra $\algA$ has pp-definitions of length $O(f(n))$, if and only if the family of algebras $\HSfactors(\algA)$ has (multi-sorted) pp-definitions of length $O(f(n))$.    
\end{lemma}

\begin{proof}
    Let us assume without loss of generality that $|A|>1$ (otherwise the statement trivially holds). For every factor $\algB_i\in\HSfactors(\algA)$ we have $\algB_i = \algS_i/\alpha_i$ for some subalgebra $\algS_i \leq \algA$ and $\alpha_i\in\Con(\algS_i)$. Note that the binary (two-sorted) relation 
    $$
    P_{\algB_i}=\{(a,a/\alpha_i) \colon a \in \algS_i\} \leq \algS_i \times \algB_i
    $$
    is invariant under $\HSfactors(\algA)$; it describes the graph of the homomorphism given by the canonical projection $\pi_{\alpha_i}\colon\algS_i \to \algB_i$ defined by $\pi_{\alpha_i}(a)=a/\alpha_i$.

    First, we will prove the direct implication. Let $\Gamma$ be a relational basis of $\Inv(\algA)$. We construct the multi-sorted relational basis $\Gamma'$ as:
    $$
    \Gamma'=\Gamma\cup\{P_{\algB_i}\mid \algB_i\in\HSfactors(\algA)\}
    $$
    
    Now, let $R' \leq \algB_1 \times \ldots \times \algB_n$ with $\algB_i \in \HSfactors(\algA)$. It is not hard to see that the (single-sorted) relation
    $$
    R = \{ (a_1,\ldots,a_n) \in S_1 \times \cdots \times S_n \mid (a_1/\alpha_1,\ldots,a_n/\alpha_n) \in R' \}
    $$
    is invariant under $\algA$. By assumption, $R$ has a pp-definition $\phi(x_1,\ldots,x_n)$ of length $O(f(n))$. The relation $R'$ can then be defined from $\Gamma'$ by the pp-formula 
    $$
    \phi'(x_1,\ldots,x_n)\ \equiv\ \exists y_1,\ldots,y_n \left(\bigwedge_{i = 1}^n P_{\algB_i}(y_i,x_i) \land \phi(y_1,\ldots,y_n)\right)
    $$
    of length $O(f(n)+2n) = O(f(n))$.
    
    For the converse, take any (multi-sorted) relational basis $\Gamma'$ of $\Inv(\HSfactors(\algA))$. Consider any relation $R'\in\Gamma'$. As above, $R' \leq \algB_1 \times \ldots \times \algB_n$ with $\algB_i \in \HSfactors(\algA)$ and the (single-sorted) relation $R = \{ (a_1,\ldots,a_n) \in S_1 \times \cdots \times S_n \mid (a_1/\alpha_1,\ldots,a_n/\alpha_n) \in R' \}$ is invariant under $\algA$.
        
    It is not hard to see that $\Gamma = \{R \mid R'\in \Gamma'\}$ is a relational basis of $\Inv(\algA)$, and for any (multi-sorted) pp-definition $\phi'$ of a relation $Q \leq \algA^n$ over $\Gamma'$, the formula $\phi$ obtained by replacing every occurrence of the symbol $R'$ by $R$ defines $Q$ over $\Gamma$.
\end{proof}

In particular, Lemma \ref{lemma:HS} implies that $\algA$ has short pp-definitions, if and only if $\HSfactors(\algA)$ has (multi-sorted) short pp-definitions. We remark that Lemma \ref{lemma:HS} does \emph{not} imply that any \emph{single} factor $\algB \in \HSfactors(\algA)$ (apart from $\algA$ itself) has short pp-definitions if $\algA$ does. In fact, we do not know if this is true (see Question \ref{question:interpretation} in the Discussion section).

\subsection{Compact representations}

For an algebra (or family of algebras) with a $k$-edge term, we know that all $n$-ary invariant relations have generating sets whose size is polynomial in $n$. In fact, we know that there are always canonical such generating sets that are called \emph{compact representations}. Compact representations were first introduced for Mal'tsev algebras in \cite{bulatov_simple_2006}, and then generalized to algebras with $k$-edge term in \cite[Definition 3.2]{berman_varieties_2010}. We are going to define a multi-sorted variant of compact representation similar to \cite[Definition 2.2.10]{brady_notes_2025}.

A \emph{fork} in a relation $R\subseteq A_1\times A_2 \times \ldots \times A_n$ is a triple $(i,a,b)$ where $i\in[n]$, $a,b\in A_i$, and there exists a pair of tuples $(\bar x,\bar y) \in R^2$ (called \emph{witnesses} of the fork) with $\proj_{[i-1]}\bar x=\proj_{[i-1]}\bar y$ and $x_i = a, y_i = b$. We define the \emph{signature} of $R$, denoted by $\Sig(R)$, as set of all of the forks $(i,a,b)$ of $R$.

Next, let $\mathcal F$ be a finite set of finite algebras that has a common $k$-edge term. Let $R$ be an $n$-ary relation that is invariant under $\mathcal F$, so $R \leq \algA_1 \times \algA_2 \times \ldots \times \algA_n$ with $\algA_i \in \mathcal F$ for all $i$. A \emph{compact representation} of $R$ is then any set of tuples $G\subseteq R$ such that:

\begin{itemize}
    \item $\Sig(G)=\Sig(R)$,
    \item $\proj_I G=\proj_I R$ for all $I\subseteq[n],|I|<k$, and
    \item $|G|\leq 2|\Sig R|+\sum_{I\subseteq[n],|I|<k}|\proj_I R|$.
\end{itemize}

Note that $|\Sig(R)|\leq c^2 n \in O(n)$ and $|G|\leq 2c^2n+ c^{k-1} \binom{n}{k-1} \in O(n^{\max(2,k-1)})$, where $c$ is an upper bound on the domain sizes in $\mathcal F$. Thus compact representations are indeed of polynomial size in $n$. It is also not hard to see that every relation $R$ has a compact representation.

We remark that our definition of compact representation slightly differs from \cite[Definition 2.2.10]{brady_notes_2025} (or \cite[Definition 3.2]{berman_varieties_2010}), as in the definition of $\Sig(R)$ we also include forks that do not come from so called \emph{minority pairs}. This is however only a minor technical detail that is not important for our analysis. In particular, in the very same way as in \cite[Theorem 2.2.11]{brady_notes_2025} (or \cite[Corollary 3.9.]{berman_varieties_2010}), one can prove that compact representations are generating sets of invariant relations:

\begin{theorem}\label{theorem:compact-representations-generate}
    Let $\mathcal F$ be a finite set of finite algebras with a common $k$-edge term. If $R\leq \algA_1 \times \ldots \times \algA_n$, where $\algA_i \in \mathcal F$, and $G$ is a compact representation of $R$, then $R=\Sg_{\algA_1 \times \ldots \times \algA_n}G$. 
\end{theorem}

We are only going to need the following corollary of Theorem \ref{theorem:compact-representations-generate}, in order to prove Lemma~\ref{lemma:critical}.

\begin{corollary}\label{corollary:subpower-with-same-forks}
    Let $\mathcal F$ be a finite set of finite algebras with a common $k$-edge term. If $R \leq R' \leq \algA_1 \times \ldots \times \algA_n$ are such that $\Sig(R)=\Sig(R')$ and $\proj_I R=\proj_I R'$ for all $I\subseteq[n]$ with $|I|<k$, then $R=R'$. 
\end{corollary}
\begin{proof}
Any compact representation $G$ of $R$ is already a compact representation of $R'$ so by Theorem~\ref{theorem:compact-representations-generate}, $R=\Sg_{\alg A^n}G=R'$.
\end{proof}

 We remark that an analogous statement does not hold in the absence of a $k$-edge polymorphism, unlike we falsely claimed for relations $R$ with the parallelogram property in \cite[Observation 16]{bulin_short_2023} of the conference version of this paper; as a counterexample take the ternary relations $R=\{(0,0,0),(1,1,1),(2,2,2),(0,1,2)\}$ and $R'=R\cup \{(0,1,1)\}$. It is easy to see that $R$ has the parallelogram property, and $\Sig(R) = \Sig(R')$ - nevertheless $R \neq R'$.

Let us mention that compact representations generalize the echelon form of a matrix, strong generating sets in permutation groups \cite{sims_computational_1970}, and that similar notions were first developed for Mal'tsev algebras~\cite{bulatov_simple_2006} (where the projections are not needed) and generalized majority-minority~\cite{dalmau_generalized_2006}. For more details, we refer the reader to \cite[Section 3]{berman_varieties_2010} and \cite[Section 2.2]{brady_notes_2025}.

\section{Main result} \label{section:main-result}

In this section, we prove the main result of our paper. We first need to introduce some standard definitions that found prominent use in the universal algebraic approach to constraint satisfaction before (see, e.g.,~\cite[Chapter 2]{brady_notes_2025}).

A binary relation $R \subseteq A \times B$ has the \emph{parallelogram property} if $(a,c),(a,d),(b,c) \in R$ implies $(b,d) \in R$. Following \cite{kearnes_clones_2012} we further say that an $n$-ary relation $R \subseteq A_1 \times A_2 \times \ldots \times A_n$ has the \emph{parallelogram property}, if for all proper subsets $I \subset [n]$ it has the binary parallelogram property when considered as a binary relation $R \subseteq (\prod_{i \in I}A_i) \times (\prod_{j \notin I}A_j)$. Note that the parallelogram property is preserved under adding dummy variables and intersection (thus, under conjunction).

An invariant relation $R \leq \algA_1 \times \ldots \times \algA_n$ is called \emph{critical} if it is \emph{$\land$-irreducible}, i.e., it cannot be written as the intersection of strictly bigger relations $Q \leq \algA_1 \times \ldots \times \algA_n$, and it has \emph{no dummy variables}.

We will make extensive use of the following property of critical relations, which follows from~\cite[Theorem 3.6]{kearnes_clones_2012} (together with \cite[Theorem 3.3]{kearnes_clones_2012}): 

\begin{theorem} \label{theorem:decompose-to-parallelogram}
Let $\algA$ be a finite algebra with $k$-edge term, and let $R \leq \algA_1 \times \ldots \times \algA_n$ be a critical relation for $\algA_1,\ldots,\algA_n \in \HSfactors(\algA)$. If $n\geq k$, then $R$ has the parallelogram property.
\end{theorem}

In particular, this implies that every invariant relation $R \leq \algA_1 \times \ldots \times \algA_n$ can be written as the conjunction $R = R' \land \bigwedge_{I \subseteq [n], |I| < k} \proj_I(R)$, where $R'$ has the parallelogram property (as also observed in~\cite[Corollary 2.3.5.]{brady_notes_2025}).

Using these notions, we can reduce the problem of finding short pp-definitions to critical relations with the parallelogram property:
                                    
\begin{lemma} \label{lemma:critical}
    Let $\algA$ be an algebra with a $k$-edge term. If for all $\algA_1,\ldots,\algA_n\in \HSfactors(\algA)$, all the critical relations $R \leq \algA_1 \times \ldots \times \algA_n$ with the parallelogram property have pp-definitions of length $O(f(n))$, then $\algA$ has pp-definitions of length $O(n^{k-1} + n\cdot f(n))$.
\end{lemma}

\begin{proof}
    Clearly, every relation $R \leq \algA_1 \times \ldots \times \algA_n$ is the intersection of the $\land$-irreducible invariant relations containing it, thus it can be written as a conjunction of critical relations. Hence, we only need to give a linear upper bound on the number of such critical relations.

    By Theorem~\ref{theorem:decompose-to-parallelogram}, we can write $R = R' \land P$, where $P =\bigwedge_{I \subseteq [n], |I| < k} \proj_I(R)$, and $R' = \bigwedge_{i=1}^K S_i$ for a set $\{S_1,S_2,\dots,S_K\}$ of $\land$-irreducible relations with the parallelogram property (such that, modulo its dummy variables, every relation $S_i$ is a critical relation of arity $\geq k$).

Then, let us consider the following chain of relations:
    $$
    P\cap S_1\geq P\cap S_1\cap S_2\geq\dots\geq P\cap\bigcap_{i=1}^j S_i\geq\dots\geq P\cap\bigcap_{i=1}^K S =P\cap R'=R
    $$
  By definition of $P$, all these relations have the same projections to $<k$-element subsets of coordinates. Note that for consecutive elements in the chain $C_j=P\cap\bigcap_{i=1}^jS_i$ and $C_{j+1}=P\cap\bigcap_{i=1}^{j+1}S_i$, we have $\Sig(C_j)\supseteq\Sig(C_{j+1})$. But this inclusion can only be strict for at most $|\Sig (P\cap S_1)|-|\Sig(R)|\leq n \cdot |A|^2$ elements. If $\Sig(C_j)=\Sig(C_{j+1})$, then by Corollary~\ref{corollary:subpower-with-same-forks} we have that $C_j=C_{j+1}$ and the critical relation corresponding to $S_{j+1}$ can be omitted from the conjunction defining $R$.
    
    Consequently, $R'$ can be defined as a conjunction of at most $n \cdot |A|^2$ many critical relations with the parallelogram property, which concludes the proof.
\end{proof}

When dealing with multi-sorted relations $R$ over $\HSfactors(\algA)$, we can furthermore always restrict the domain $\algA_i$ of the $i$-th variable of a relation $R$ to its projection $\proj_i(R) \leq \algA_i$. So, without loss of generality, we can assume that $R \leq_{sd} \algA_1 \times \ldots \times \algA_n$ is \emph{subdirect}, i.e., its projection to every coordinate $i$ is the full domain $\algA_i$.

For a subdirect relation $R \leq_{sd} \algB \times \algC$ with the parallelogram property, let us define the \emph{linkedness congruence $\theta_B$} on $\algB$ as follows:
$$
(x,y) \in \theta_B\ \leftrightarrow\ (\exists c \in C)(R(x,c) \land R(y,c))
$$

For a general relation $R \leq_{sd} \algA_1 \times \ldots \times \algA_n$ with the parallelogram property, and any proper subset $I \subset [n]$ of coordinates, we define the \emph{linkedness congruence $\theta_I$} on $\proj_I(R)$ analogously, where we consider $R$ as a binary relation between $\proj_I(R)$ and $\proj_{[n]\setminus I}(R)$. We will also write $\theta_i$ instead of $\theta_{\{i\}}$ (in \cite{kearnes_clones_2012}, this relation was referred to as $i$-th coordinate kernel). It follows from the parallelogram property that $\theta_I$ is indeed a congruence of $\proj_I(R)$. 

A subdirect relation $R \leq_{sd} \algA_1 \times \ldots \times \algA_n$ is \emph{reduced} if 
every tuple $(a_1,\ldots,a_n) \in R$ is already uniquely determined by $(a_1,\ldots,a_{i-1},a_{i+1},\ldots,a_n)$, for any coordinate $i$; in other words, $\theta_i$ is trivial, for every $i = 1,\ldots,n$. The following lemma narrows down the quest for short pp-definitions to \emph{reduced}, subdirect, critical relations with the parallelogram property:

\begin{lemma} \label{lemma:reduced}
    Let $\algA$ be an algebra with a $k$-edge term.  Assume that all relations $R \leq_{sd} \algA_1 \times \cdots \times \algA_n$ with $\algA_i \in \HSfactors(\algA)$ that are reduced, critical, and have the parallelogram property, have (multi-sorted) pp-definitions of length $O(f(n))$. Then $\algA$ has pp-definitions of length $O(n^{k-1} + n \cdot f(n))$.
\end{lemma}

\begin{proof}
    We first prove that we can pp-define all critical relations $R \leq_{sd} \algA_1 \times \ldots \times \algA_n$ with $\algA_i \in \HSfactors(\algA)$ that have the parallelogram property (but are not necessarily reduced), by pp-definitions of length at most $O(f(n))$. Given such a relation, let us consider the linkedness congruence $\theta_{i} \in \Con(\algA_i)$ for every coordinate $i$. Then let us define the quotient $R' = R/(\theta_1,\ldots,\theta_n)$. It is not hard to see that $R' \leq_{sd} \algA_1/\theta_1 \times \cdots \times \algA_n/\theta_n$ is also critical, and has the parallelogram property. Furthermore, by definition of the linkedness congruence $\theta_i$, $R'$ is reduced. 
    
    Since $R$ has the parallelogram property, it is equal to the full preimage of $R'$ under the quotient map $(x_1,\ldots,x_n) \mapsto (x_1/\theta_1,\ldots,x_n/\theta_n)$. Thus, every pp-definition $\phi'(x_1,\ldots,x_n)$ of $R'$ gives rise to the pp-definition 
    $$
    \phi(x_1,\ldots,x_n) \equiv \exists y_1,\ldots,y_n \left( \bigwedge_{i = 1}^n (x_i/\theta_i = y_i) \land \phi'(y_1,\ldots,y_n) \right)$$ 
    of length $O(f(n))$ which defines $R$; this proves our claim. The statement of the lemma then follows immediately from (the multi-sorted version of) Lemma \ref{lemma:critical} together with Lemma \ref{lemma:HS}.
\end{proof}

Any relation $R \leq_{sd} \algA_1 \times \ldots \times \algA_n$ as in Lemma \ref{lemma:reduced} comes with several nice properties. In algebraic terms, it is a \emph{graph of a joint similarity} between the algebras $\algA_i$, see \cite{kearnes_clones_2012} (or also \cite[Section 2.3.1]{brady_notes_2025}). For our proof, we are only going to need the result presented in Lemma \ref{lemma:SI}, or rather its generalization in Lemma \ref{lemma:linkedness}:

\begin{lemma}[{\cite[Lemma 2.4]{kearnes_clones_2012}}] \label{lemma:SI}
    Let $\algA_1,\ldots,\algA_n$ be algebras with few subpowers, and let $R \leq_{sd} \algA_1 \times \ldots \times \algA_n$ be a reduced, critical relation with the parallelogram property. Then every $\algA_i$ is subdirectly irreducible.
\end{lemma}

\begin{lemma} \label{lemma:linkedness}
    Let $\algA_1,\ldots,\algA_n$ be algebras with few subpowers, and let $R \leq_{sd} \algA_1 \times \ldots \times \algA_n$ be a critical relation with the parallelogram property. For $I \subset [n]$, let $\theta_I$ be the linkedness congruence on $\proj_I(R)$ with respect to $R$. Then $\theta_I$ is $\land$-irreducible.
\end{lemma}

\begin{proof} 
    Since $R$ is $\land$-irreducible, there is a unique cover $R^* > R$ in the lattice of all subalgebras of $\algA_1 \times \ldots \times \algA_n$. Any tuple $\bar a = (a_1,\ldots,a_n) \in R^* \setminus R$ is called a \emph{key tuple} of $R$ (cf.~\cite{zhuk_key_2017} and \cite[Lemma 2.1]{kearnes_clones_2012}). Note that for every $j = 1, \ldots, n$, there is a tuple $(a_1,\ldots b_j, \ldots ,a_n) \in R$ that only differs from $\bar a$ at position $j$. If this was not the case, say for $j=n$, then $\bar a\notin\proj_{[n-1]}R\times A_n$. Since $R\subseteq(\proj_{[n-1]}R\times A_n) \cap R^*\subsetneq R^*$, we get that $R=(\proj_{[n-1]}R\times A_n) \cap R^*$. This contradicts the criticality of $R$: either it has a dummy variable or it is an intersection of strictly bigger relations.
        
    For simplicity, let us assume that $I = \{1,2,\ldots,i\}$ with $i < n$. Then, the linkedness congruence $\theta_I$ has an equivalence class containing all elements of the form  $(a_1,\ldots b_j, \ldots ,a_i)$ for $j = 1,\ldots,i$. Note that $(a_1,a_2,\ldots,a_i) \in \proj_I(R)$ is not an element of this equivalence class.
    
    To prove that $\theta_I$ is $\land$-irreducible, let $\theta'$ be a congruence strictly above $\theta_I$. We claim that then $\theta'$ must also contain the pair $((a_1,a_2,\ldots,a_i), (b_1,a_2,\ldots,a_i))$. To prove the claim, consider the relation $R'$ given by the following pp-definition:
    $$
    \exists \bar y \left( \theta'(\proj_I \bar x,\bar y) \land R(\bar y, \proj_{[n]\setminus I} \bar x) \right).
    $$    
    As $R'$ properly contains $R$, it also must contain its cover $R^*$, and thus the key tuple $(a_1,a_2,\ldots,a_n)$. Moreover, the linkedness congruence of $R'$ on coordinates $I$ is equal to $\theta'$, thus $\theta'$ must contain the pair $((a_1,a_2,\ldots,a_i), (b_1,a_2,\ldots,a_i))$. So $\theta_I$ has a unique cover $\theta_I^*$, which is the congruence generated by $\theta_I\cup\{((a_1,a_2,\ldots,a_i), (b_1,a_2,\ldots,a_i))\}$. This finishes the proof.
\end{proof}

We are now ready to prove our main result:

\begin{theorem}\label{theorem:main-result}
    Let $\algA$ be an algebra with a $k$-edge term, and assume that $\HSP(\algA)$ is residually finite. Then $\algA$ has pp-definitions of length $O(n^{\max(2,k-1)})$.
\end{theorem}

\begin{proof}
    Let $\mathcal V_{SI}$ consist of all subdirectly irreducible elements of $\HSP(\algA)$. Since $\HSP(\algA)$ is residually finite, $\mathcal V_{SI}$ contains up to isomorphism only finitely many algebras, all of which are finite. In particular, all members of $\mathcal V_{SI}$ are already contained (up to isomorphism) in $\HSfactors(\algA^l)$, for some finite power $l$.
    
    By Lemma \ref{lemma:poweralg}, it is enough to prove pp-definitions of length $O(n^{k})$ for $\algA^l$. By Lemma~\ref{lemma:reduced}, it suffices to prove that every reduced, critical relation $R \leq_{sd} \algA_1 \times \ldots \times \algA_n$ with $\algA_i \in \HSfactors(\algA^l)$
    that has the parallelogram property, has a (multi-sorted) pp-definition of linear length.
    
    Let $\Gamma$ be the set of all at most ternary invariant relations over $\HSfactors(\algA^l)$. We construct a pp-definition of $R$ from $\Gamma$ of length linear in $n$, by induction on $n$. For $n \leq 3$, $R$ itself is in~$\Gamma$.
    
    For general $R \leq_{sd} \algA_1 \times \ldots \times \algA_n$ with the parallelogram property, recall the definition of the linkedness congruence $\theta_I$. We then define the algebra $\algA_{1,2} = \proj_{\{1,2\}}(R)/\theta_{\{1,2\}}$, and the following relations:
    \begin{align*}
        Q &= \{(x_1,x_2,y_{1,2}) \in A_1 \times A_2 \times A_{1,2} \mid y_{1,2} = (x_1,x_2)/\theta_{{1,2}} \}\\
        R' &= \{(y_{1,2},x_3,\ldots,x_n) \mid \exists x_1,x_2 \left( Q(x_1,x_2,y_{1,2}) \land R(x_1,x_2,x_3, \ldots, x_n) \right) \}
    \end{align*}
    Note that $Q \leq \algA_1 \times \algA_{2} \times \algA_{1,2}$, and $R' \leq \algA_{1,2} \times \algA_3 \times \cdots \times \algA_n$. Since $R$ has the parallelogram property, $R$ can be defined by the following pp-formula (see Figure~\ref{figure:similarity-construction}): 
    $$
    (\exists y_{1,2} \in \algA_{1,2}) \left( Q(x_1,x_2,y_{1,2}) \land R'(y_{1,2},x_3,\ldots,x_n) \right)
    $$

    \begin{figure}
        \centering
        \begin{tikzpicture}[scale=1.75]
            \filldraw [color=black!65] (2.9,-1) circle (1pt) node[left]{\large $y_{1,2}$};
            \draw [rounded corners=1em, color=black!65](2.4,0.2) rectangle (3.4,-2.2); 
            \draw (2.9,0.5) node {$\algA_{1,2}$};   
            \draw [line width=2pt, dotted, color=black!65] plot [smooth, tension=1.75] coordinates {(0.5,-0.5) (1.6,-0.5) (2.9,-1)};
            \draw [line width=2pt, dotted, color=black!65] plot [smooth, tension=1.75] coordinates {(0.5,-1.5) (1.6,-1.5) (2.9,-1)};
            
            \draw [rounded corners=1em](0,0) rectangle (1,-2); 
            \draw [rounded corners=1em](1.2,0) rectangle (2.2,-2); 
            \draw [rounded corners=1em](3.6,0) rectangle (4.6,-2); 
            \draw [rounded corners=1em](5.6,0) rectangle (6.6,-2); 
        
            \filldraw (4.2,-1) circle (1pt) node[above]{$x_3$};
            \filldraw (6.2,-1) circle (1pt) node[above]{$x_n$};
            \draw [thick, decorate, decoration={snake,amplitude=0.8}] (4.2,-1) -- node[above]{$\cdots$} (6.2,-1);
        
            \filldraw (0.5,-0.5) circle (1pt) node[above]{$x_1$};
            \filldraw (1.6,-0.5) circle (1pt) node[above]{$x_2$};
            \draw [thick] (0.5,-0.5) -- (1.6,-0.5);
        
            \filldraw (0.5,-1.5) circle (1pt) node[below]{$x_1'$};
            \filldraw (1.6,-1.5) circle (1pt) node[below]{$x_2'$};
            \draw [thick] (0.5,-1.5) -- (1.6,-1.5);    
        
            \draw [thick] (1.6,-1.5) -- (4.2,-1);
            \draw [thick] (1.6,-0.5) -- (4.2,-1); 
        
            \draw[thick, decorate, decoration={brace,amplitude=5}] (3.4,-2.3) -- node[below=5pt] {$Q$} (0,-2.3);  
            \draw [line width=2pt, dotted, color=darkgray] (2.9,-1) -- (4.2,-1); 
            \draw[thick, decorate,decoration={brace,amplitude=5}] (6.6,-2.5) -- node[below=5pt] {$R'$} (2.4,-2.5);
        \end{tikzpicture}

        \caption{The construction from the proof of Theorem~\ref{theorem:main-result}}
        \label{figure:similarity-construction}
    \end{figure}
    
    By Lemma \ref{lemma:linkedness}, the linkedness congruence $\theta_{\{1,2\}}$ is $\land$-irreducible. This implies that $\algA_{1,2}$ is subdirectly irreducible and hence an element of $\mathcal V_{SI} \subseteq \HSfactors(\algA^l)$. In particular, this means that $Q$ is a relation from our relational basis $\Gamma$. The relation $R'$ is of arity $n-1$, and thus, by induction assumption, has a pp-definition of linear length. This finishes our proof.
\end{proof}

Note that, although in the proof of Theorem \ref{theorem:main-result} we found a \emph{ternary} multi-sorted constraint language $\Gamma$ defining the reduced critical relations $R \leq_{sd} \algA_1 \times \ldots \times \algA_n$ with the parallelogram property, the same may not be true for the original algebra~$\algA$. Our construction only leads to an upper bound of $3\ell$ on the arity, where $\ell$ is such that all subdirectly irreducible elements of $\HSP(\algA)$ are, up to isomorphism, contained in $\HSfactors(\algA^\ell)$. In the general case, we are not aware of any better bounds for $\ell$ than the double exponential bound $\ell \leq |A|^{|A|^{|A|+1}+1}$ provided in \cite[Theorem 10.15]{freese_commutator_1987} (see also \cite[Theorem A.5.27]{brady_notes_2025}). But, we remark that in some cases much better lower bounds on the arity of relational bases are known (e.g. for Mal'tsev algebras on 3-elements \cite{bulatov_three_2005}).

As an immediate consequence of Theorem \ref{theorem:main-result} we get that every 3-element algebra with few subpowers has short pp-definitions, confirming Conjecture~\ref{conjecture:short-definitions-iff-few-subpowers} for the 3-element case. 

\begin{corollary} \label{cor:3}
    Let $\Gamma$ be a constraint language on a 3-element domain. Then $\Gamma$ has short pp-definitions if and only if $\Gamma$ has few subpowers. More precisely, $\Gamma$ has pp-definitions of length $O(n^{\max(2,k-1)})$, where $k$ is the minimal number such that $\Gamma$ has a $k$-edge polymorphism.
\end{corollary}

\begin{proof}
    Let us assume that $\Gamma$ is a constraint language with a $k$-edge polymorphism, and let $\algA = (A;\Pol(\Gamma))$ be its polymorphism algebra. The existence of an edge polymorphism implies that $\HSP(\algA)$ is congruence modular \cite[Theorem 4.2]{berman_varieties_2010}. It is known that every algebra with at most 3 elements in a congruence modular variety generates a residually small subvariety: This fact follows from the Freese–McKenzie characterization of finitely generated, residually small, congruence modular varieties in \cite{freese_residually_1989}, for a proof see e.g. \cite[Corollary A.5.31.]{brady_notes_2025}. Thus $\HSP(\algA)$ is residually finite. By Theorem \ref{theorem:main-result}, $\algA$ (and therefore also $\Gamma$) has pp-definitions of length $O(n^{\max(2,k-1)})$.

    If $\Gamma$ has few subpowers, then by Theorem \ref{theorem:few_subpowers_iff_edge_term}, it has a $k$-edge polymorphism for some $k$. Thus $\Gamma$ has short pp-definitions if it has few subpowers.
\end{proof}

We furthermore know (see e.g. \cite{freese_residually_1989}) that finite algebras that are either affine or have an NU term, generate residually finite varieties. Thus, we also get Theorem \ref{theorem:affine} and Theorem \ref{theorem:NU} as direct corollaries of Theorem \ref{theorem:main-result}.

\section{The subpower membership problem}\label{section:SMP}

The \emph{Subpower Membership Problem $\SMP(\algA)$} of a finite algebra $\algA$ is the computational problem in which the input consists of a list of tuples $\bar b, \bar a_1, \ldots, \bar a_k \in A^n$, for arbitrary $n \geq 1$, and one needs to decide whether $\bar b$ lies in the subalgebra $\Sg_{\algA^n}\{\bar a_1, \ldots, \bar a_k\}$ generated by $\bar a_1, \ldots, \bar a_k$, i.e., in the smallest $R \leq \algA^n$ that contains all the tuples $\bar a_1, \ldots, \bar a_k$.

The existence of an efficient algorithm for $\SMP(\algA)$ implies that it is feasible to represent the relations in $\Inv(\algA)$ by some generating set of tuples. In particular, in the context of constraint satisfaction, it was remarked by several authors (see, e.g.,~\cite{bulatov_subpower_2019}) that a polynomial-time algorithm for $\SMP(\algA)$ would allow us to define constraint satisfaction problems over \emph{infinite} constraint languages $\Gamma \subseteq \Inv(\algA)$, where the constraint relations in $\Gamma$ are encoded by generating tuples. In \cite{idziak_tractability_2010}, any algebra $\algA$ with $\SMP(\algA)$ in P was referred to as \emph{polynomially evaluable}.

While there are algebras for which $\SMP(\algA)$ is EXPTIME-complete \cite{kozik_finite_2008}, it was asked in \cite[Question 3]{idziak_tractability_2010} whether all algebras with few subpowers are polynomially evaluable. An affirmative answer was given for several special cases \cite{mayr_subpower_2012, bulatov_subpower_2019}, but the question still remains open in general. 

The best general upper bound on the complexity of $\SMP(\algA)$ for algebras with few subpowers is $\NP$ \cite{bulatov_subpower_2019}. This bound is based on the fact that, for a tuple $\bar b$, and a compact representation $G$ of $R \leq \algA^n$, we can verify in polynomial time whether $\bar b$ is generated by $G$. Thus, compact representations give rise to polynomial-time witnesses for `Yes'-instances of $\SMP(\algA)$. The difficulty in finding \emph{deterministic} polynomial algorithms for $\SMP(\algA)$ lies in efficiently computing such compact representations $G$ from an arbitrary generating set $\{\bar a_1, \ldots, \bar a_k\}$ of $R$.

Note that for an algebra $\algA$ with $\Inv(\algA) = \langle \Gamma \rangle$, the \emph{non}-membership of a tuple $\bar b$ in a relation $R=\Sg_{\algA^n}\{\bar a_1, \ldots, \bar a_k\}$ can be certified by a pp-formula $\phi(\bar x)$ over $\Gamma$, such that $\phi$ holds for all tuples $\bar a_1, \ldots, \bar a_k$, but not for $\bar b$.

If $\Gamma$ has short pp-definitions, we can guess such a certificate $\phi$ for `No'-instances of $\SMP(\algA)$, and verify it in polynomial time. (To elaborate, we can consider $\phi$ as an instance of $\CSP(\Gamma)$, fix the values of the free variables to the corresponding values from the given tuple, and verify that the CSP algorithm \cite{bulatov_dichotomy_2017,zhuk_proof_2017}) answers `yes' for the tuples $\bar a_i$ but `no' if the tuple is $\bar b$.)

As a direct consequence of this (and the fact that short pp-definitions imply few subpowers), we obtain the following: 

\begin{theorem} \label{theorem:SMP}
    Let $\algA$ be an algebra with short pp-definitions. Then $\SMP(\algA) \in \NP \cap \coNP$.
\end{theorem}

In particular, an affirmative answer to Conjecture~\ref{conjecture:short-definitions-iff-few-subpowers} would imply that the problem $\SMP(\algA)$ is in $\NP \cap \coNP$ for every algebra $\algA$ with few subpowers. Moreover, if one could efficiently compute the certificate, i.e., construct a short pp-definition of a relation from its generators, we would have even a polynomial-time algorithm for $\SMP(\algA)$; see Question~\ref{question:ppp} below.

Let us note, however, that in the setting of our main result (Theorem \ref{theorem:main-result}), this approach does not provide any progress on the state-of-art of the subpower membership problem. Indeed, it was already shown in \cite{bulatov_subpower_2019} that $\SMP(\algA)$ is in P for every algebra $\algA$ with few subpowers that generates a residually finite variety.

\section{Discussion} \label{section:discussion}

In this section, we propose several possible directions of further research towards expanding the scope of Theorem~\ref{theorem:main-result} beyond the property of residual finiteness, as well as strengthening its statement.

\subsection{Beyond residual finiteness}

By Theorem \ref{theorem:main-result}, constraint languages with few subpowers, whose polymorphism algebra generates a residually finite variety, have short pp-definitions. While this confirms Conjecture~\ref{conjecture:short-definitions-iff-few-subpowers} for a large subclass of constraint languages, much work remains to prove the conjecture in full generality. As discussed in Section \ref{section:SMP}, it is also essential for progress on the Subpower Membership Problem to extend our results to algebras that do not generate residually finite varieties.

While we do not present any general results beyond the context of residual finiteness in this paper, we are aware of singular such examples of algebras that generate non-residually finite varieties but still have short pp-definitions thus providing some further evidence in favor of Conjecture~\ref{conjecture:short-definitions-iff-few-subpowers}. Let us give one such example here, a 4-element algebra that was described by Brady in \cite[Example 2.3.2.]{brady_notes_2025}.

\begin{example} \label{example:zeb}
Let $U=L=\{0,1\}$, and let $\alg A=(U\times L;g)$ where $g$ is defined by 
$$g \left(\begin{pmatrix} u_1 \\ \ell_1\end{pmatrix}, \begin{pmatrix} u_2 \\ \ell_2\end{pmatrix}, \begin{pmatrix} u_3 \\ \ell_3\end{pmatrix} \right) = \begin{pmatrix} \mathrm{maj}(u_1,u_2,u_3) \\ \ell_1+\ell_2+\ell_3 + \hat g(u_1,u_2,u_3) \end{pmatrix},$$
where $+$ denotes the addition modulo 2 and $\hat g\colon U^3\to L$ is a symmetric operation with $\hat g(0,1,1) = 1$ and $\hat g(u_1,u_2,u_3) = 0$ else. Then $\HSP(\algA)$ is not residually finite, and $\algA$ has pp-definitions of length $O(n^2)$.
\end{example}

\begin{proof}
Note that by the map 
$$a \mapsto \binom00, b \mapsto \binom10, c \mapsto \binom01, d \mapsto \binom11,$$
 $\algA$ is an isomorphic copy of the algebra $\mathbb A$ described in \cite[Example 2.3.2.]{brady_notes_2025}. The fact that $\algA$ has pp-definitions of length $O(n^2)$ can already be seen from a very careful read of \cite[Example 2.3.2.]{brady_notes_2025}, but we will discuss the example in more detail for the reader's convenience. We first remark that 
\begin{itemize}
\item $\algA$ has a 3-edge term defined by $$e(u, x, y, z) = g(x, g(u, y, y), g(y, g(x, y, z), g(x, y, z))).$$
\item The map $$\sigma \binom u \ell = \binom{u+1}{\ell+u}$$ is an automorphism of $\algA$ (thus its function graph is invariant under $\algA$).
\item The only non-trivial congruence $\mu$ of $\algA$ is the kernel of the projection to $U$.
\item The congruence $\mu$ is central (in the sense of commutator theory), but $\algA$ is not affine. Therefore, by the characterization in \cite[Theorem 10.15]{freese_commutator_1987}, $\algA$ does not generate a residually finite variety.
\end{itemize}
It is furthermore not hard to see that the set of factors $\HSfactors(\algA)$ consists only of $\algA$, the Boolean majority algebra $\alg{U} = \algA/\mu$ and the two Abelian subalgebras $\alg{L}_0,\alg{L}_1$ on domains $\{0\} \times L$ and $\{1\} \times L$, which are both isomorphic to $(L, x+y+z)$. In particular, the graph of the isomorphism $\algL_1 \cong \algL_2$ is a binary invariant relation under $\HSfactors(\algA)$. In the light of Lemma \ref{lemma:reduced}, it is enough to prove that any reduced critical relation $R \leq_{sd} \algA^{m} \times \algL_0^{r_0} \times \algL_1^{r_1} \times \algU^s$ of arity $n = m+r_0+r_1 +s \geq 3$ with the parallelogram property has a pp-definition of length $O(n)$. So, in the following let $R$ be such a relation.

We know that $s=0$; this follows from \cite[Theorem 2.5 (7)]{kearnes_clones_2012}, and the fact that $\algU$ is simple and non-affine (or from a straightforward computation, using the fact that $\algU$ is the majority algebra). Thus $R \leq_{sd} \algA^{m} \times \algL_0^{r_0} \times \algL_1^{r_1}$. By composing $R$ with the isomorphism that maps $\algL_1$ to $\algL_0$ in each of the last $r_1$ coordinates, we further see that $R$ is interdefinable with a relation $R' \leq_{sd} \algA^{m} \times \algL_0^{r_0+r_1}$ via a pp-formula of length $O(r_1)$. Thus without loss of generality we can assume that $r_1 = 0$, and write $r_0 = r$. Note that, if $m=0$, we can already find a linear pp-definition of $R$ from ternary relations, by Example \ref{example:lin}. In fact, the same argument also works for all relations with $m\leq 2$, so from now on let us assume that $m \geq 3$.

Instead of $R$ we can furthermore consider its compositions with powers of $\sigma$ in any of the first $m$ coordinates. Therefore, without loss of generality, we can assume that the constant $\binom00$-tuple is in $R$. 

It follows from \cite[Theorem 2.5]{kearnes_clones_2012} (see also \cite[Theorem 2.3.10]{brady_notes_2025}) that the unique cover $R^*$ of $R$ in the lattice of subalgebras of $\algA^m \times \algL_0^r$ has the parallelogram property, and that its linking congruences are $\mu$ in the first $m$ coordinates, and $1_{\algL_0}$ in the last $r$ coordinates. Thus the reduced relation $R^*/(\mu, \ldots, \mu, 1_{\algL_0}\ldots  1_{\algL_0}) \leq_{sd} \algU^m$ has the parallelogram property. By the properties of the majority algebra $\algU$ the only such relations that contain a constant tuple are those definable by a conjunction of equalities. Thus, for any two coordinates $1\leq i < j \leq m$, the projection $\proj_{i,j}R^* \leq_{sd} \algA^2$ is either equal to $\mu$ or $A^2$. Note further that $\proj_{i,j}R = \proj_{i,j}R^*$, as we could otherwise write $R$ as conjunction of $R^*$ and $\proj_{i,j}R$, implying that $R$ is not critical.

Let us write $i \sim j$ if $\proj_{i,j}R = \mu$. This clearly defines an equivalence relation on the coordinates in $[m]$. Let us further write $R \leq_{sd} \algA^{m_1} \times \algA^{m_2} \times \ldots \times \algA^{m_k} \times \algL_0^r$, where this subdivision of $[m]$ corresponds to the equivalence classes of $\sim$.

Note that, since $R$ contains the constant $\binom00$-tuple, it must also be invariant under the map $$\phi\binom{u}{\ell} =  g\left(\binom00,\binom{u}{\ell},\binom{u}{\ell}\right) = \binom{u}{u},$$ which sends $L_0$ to $\binom00$, and $L_1$ to $\binom11$. From this, and the fact that $\proj_{i,j}R = A^2$ for $i\not\sim j$, it can be derived that $R$ contains the set $P = \{ \bar x \in A^m \times \{\binom00\}^r \mid x_i = x_j \text{ if } i \sim j \}$ of piece-wise constant tuples.

Moreover, observe that $R_0 = R \cap \algL_0^{m+r}$ is a linear subspace of $\algL_0$ and, for every tuple $i = (i_1,i_2,\ldots,i_k) \in \{0,1\}^k$, every subspace $R_i = R \cap \algL_{i_1}^{m_1} \times \ldots \algL_{i_k}^{m_k} \times \algL_0^{r}$ is isomorphic to $R_0$ by the map $\psi \colon A \to L$ defined by $$\psi \binom{u}{\ell} =  g\left( \binom00, \binom00, \binom{u}{\ell} \right) = \begin{pmatrix}0 \\ \ell\end{pmatrix}.$$ We conclude that $R$ is the disjoint union of components $R_i$, which are all equal to the linear space $R_0 \leq \algL_0^{m+r}$, modulo relabeling the $U$-component in the $\sim$-blocks. Conversely, note that for every $Q_0 \leq \algL_0^{m+r}$ that contains $P\cap \algL_0^{m+r}$ we can define $Q = \bigcup_{i \in \{0,1\}^k} Q_i$, where $Q_i = \psi^{-1}(Q_0) \cap (L_{i_1}^{m_1} \times \ldots L_{i_k}^{m_k} \times L_0^{r})$, which is also an invariant relation $Q \leq \algA^m \times \algL_0^r$

Since $R$ is critical, this implies that $m_1 = \ldots = m_k = 2$, and $R_0$ is given by the equation $\sum_{i=1}^{m+r} \ell_i = 0$ (otherwise it could be written as a conjunction of relations of the form $Q$ and $\mu$ constraints). It is now not hard to see that we can construct a pp-formula of length $O(n) = O(m+r)$ of $R$ from the ternary relations $$\left\{\left(\binom{u}{\ell_1}, \binom{u}{\ell_2},\binom{0}{\ell_3}\right) \in A^3 \mid \ell_1 + \ell_2 + \ell_3 = 0 \right\},$$ and an argument as in Example \ref{example:lin}. From Lemma \ref{lemma:reduced} we get that $\algA$ has pp-definitions of length $O(n^2)$.
\end{proof}

Let us note here that the condition of residual finiteness has traditionally not played a major role in the study of constraint satisfaction problems, being more commonly encountered in purely algebraic contexts.

\subsection{Short definitions under pp-interpretability}

A natural direction for connecting short pp-definitions more closely with the theory of constraint satisfaction is to explore whether our results extend to tractability classes that figure prominently in the algebraic approach to the CSP, such as constraint languages with Mal'tsev polymorphisms. For that purpose, establishing invariance of the property under pp-interpretations seems important and could facilitate further progress.

As already mentioned, pp-interpretations are a generalization of pp-definitions, that describe gadget reductions between constraint languages on different domains. Many standard tractability classes, such as few subpowers, Mal'tsev, near-unanimity, are closed under pp-interpretations. This motivates the following question:

\begin{question} \label{question:interpretation}
    Let $\Gamma$ and $\Delta$ be two constraint languages such that $\Delta$ is pp-interpretable in $\Gamma$ and $\Gamma$ has short pp-definitions. Then, does $\Delta$ also have short pp-definitions?
\end{question}

Note that Conjecture~\ref{conjecture:short-definitions-iff-few-subpowers} implies a positive answer. We remark that we do not know the answer to this question even for $\Delta$ being \emph{pp-definable} in $\Gamma$. Only in the special case of \emph{pp-bi-interpretable} structures, Question \ref{question:interpretation} has a positive answer by Lemma \ref{lemma:ppbiinterpret}, we thank the anonymous reviewer of \cite{bulin_short_2023} for this observation.

In algebraic terms, Question \ref{question:interpretation} asks whether for $\algA$ with short pp-definitions, it is the case that also every \emph{extension} of every algebra $\algB \in \HSP^{\mathrm{fin}}(\algA)$ has short pp-definitions; for finite powers $\PPP^{\mathrm{fin}}$, we verified the statement in Lemma \ref{lemma:poweralg}.

\subsection{Efficiently computable short pp-definitions}

In order to improve the complexity result of Theorem \ref{theorem:SMP} and put $\SMP(\algA)$ in the class P, we would need an explicit method of efficiently computing a short pp-definition for a relation $R = \Sg_{\algA^n}\{\bar a_1,\ldots,\bar a_k\}$ given by its generators $\bar a_1,\ldots,\bar a_k$. This motivates the following question:

\begin{question} \label{question:ppp}
    Let $\Gamma$ be a constraint language with short pp-definitions. Is there a polynomial-time algorithm that computes a (short) pp-definition of a relation $R \leq \langle \Gamma \rangle_n$, given by a set of generators $\bar a_1,\ldots,\bar a_k$?
\end{question}

We remark that over the Boolean domain, Question \ref{question:ppp} has a positive answer (see Examples \ref{example:lin} and \ref{example:2SAT}) but we do not know the answer even in the context of residual finiteness.

\subsection{Bounding the polynomial degree}

Finally, recall that the bound from Theorem~\ref{theorem:main-result} is a polynomial of degree $\max(2,k-1)$ if $\Gamma$ has a $k$-edge polymorphism. It is therefore tempting to conjecture that the same degree could be enough in general for Conjecture~\ref{conjecture:short-definitions-iff-few-subpowers}. Note that the number of $n$-ary pp-definable relations, for $\Gamma$ with a $k$-edge polymorphism, is known to be in $2^{O(n^k)}$~\cite[Theorem 3.4]{idziak_tractability_2010}.

\section*{Acknowledgments}

The authors would like to thank Dmitriy Zhuk for inspiring discussions about critical relations and the anonymous reviewers of \cite{bulin_short_2023} for their valuable input, including the suggestion to include Lemma \ref{lemma:ppbiinterpret}.

\bibliographystyle{siamplain}
\bibliography{references}

\end{document}